\documentclass[conference]{IEEEtran}
\usepackage{amsmath,amssymb, amsthm}
\usepackage[utf8]{inputenc}
\usepackage[english]{babel}
\usepackage{cite}
\usepackage{epsfig}
\usepackage{float}
\usepackage{graphicx, caption}
\usepackage{url}
\usepackage{subfig}
\usepackage{times}
\usepackage{epstopdf}
\usepackage{amsxtra}
\usepackage{rawfonts}
\usepackage{here}
\usepackage{enumerate}
\usepackage[numbers,sort&compress]{natbib}
\newtheorem{corollary}{\bf Corollary}
\newtheorem{theorem}{\bf Theorem}

\newtheorem{lemma}{\bf Lemma}

\newtheorem{definition}{\bf Definition}

\begin{document}
\IEEEoverridecommandlockouts
\title{Breaking the Economic Barrier of Caching in Cellular Networks: Incentives and Contracts \vspace{-0.3cm}}

\author{{Kenza Hamidouche$^{1,2}$}, Walid Saad$^{2}$, and M\'erouane Debbah$^{1,3}$\vspace*{0em}\\
\authorblockA{\small $^{1}$ CentraleSup\'elec, Universit\'e Paris-Saclay, Gif-sur-Yvette,France, Email: \url{kenza.hamidouche@centralesupelec.fr}\\
$^{2}$Wireless@VT, Bradley Department of Electrical and Computer Engineering, Virginia Tech, Blacksburg, VA, USA, Email: \url{walids@vt.edu} \\
$^{3}$Mathematical and Algorithmic Sciences Lab, Huawei France R\&D, France, Email: \url{merouane.debbah@huawei.com}
\vspace{-0.5cm}}
    \thanks{This research was supported by ERC Starting Grant 305123 MORE (Advanced Mathematical Tools for Complex Network Engineering), the ANR project: WisePhy: S\'ecurit\'e pour les communications sans fil \`a la couche physique, the U.S. National Science Foundation under Grants CNS-1513697 and CNS-1460316.}%
  }

\maketitle

\begin{abstract}
In this paper, a novel approach for providing incentives for caching in small cell networks (SCNs) is proposed based on the economics framework of contract  theory. In this model, a mobile network operator (MNO) designs contracts that will be offered to a number of content providers (CPs) to motivate them to cache their content at the MNO's small base stations (SBSs). A practical model in which information about the traffic generated by the CPs' users is not known to the MNO is considered. Under such asymmetric information, the incentive contract between the MNO and each CP is properly designed so as to determine the amount of allocated storage to the CP and the charged price by the MNO. The contracts are derived by the MNO in a way to maximize the global benefit of the CPs and prevent them from using their private information to manipulate the outcome of the caching process. For this interdependent contract model, the closed-form expressions of the price and the allocated storage space to each CP are derived. This proposed mechanism is shown to satisfy the sufficient and necessary conditions for the feasibility of a contract. Moreover, it is shown that the proposed pricing model is budget balanced, enabling the MNO to cover all the caching expenses via the prices charged to the CPs. Simulation results show that none of the CPs will have an incentive to choose a contract designed for CPs with different traffic loads.
\end{abstract}
\section{Introduction}
The capacity limitations of the backhaul links that connect small base stations (SBSs) to the core network in a small cell network (SCN), make it difficult
for the mobile network operators (MNOs) and content providers (CPs) to ensure the required data rates for emerging, bandwidth intensive users' applications especially during peak hours \cite{andrews2014will}. To overcome these backhaul limitations and meet users' requirements in terms of quality-of-service (QoS), distributed caching at the network edge has been recently proposed as a new promising solution \cite{poularakisexploiting,blaszczyszyn2014optimal,blasco2014content,ji2014fundamental}. 

Caching typically relies on storing the most popular files at the levels of SBSs and devices, to reduce backhaul traffic \cite{poularakisexploiting}. To successfully deploy SCN caching solutions, MNOs require the cooperation of the CPs by sharing their content and providing this content's global \cite{paschos2016wireless}. However, although CPs can improve the QoS of their users by caching, they might be reluctant to share their content with the MNOs. This can be due to reasons such as privacy since, once the content cached at the SBSs, the MNOs can get access to all the CP's files as well as the traffic dynamics of the users subscribed to the CP \cite{paschos2016wireless}. Thus, the MNOs must provide incentives to the CPs to share their data and cache it at the SBSs by introducing suitable economic arrangements that can be beneficial for both the MNO and the CPs \cite{paschos2016wireless,li2016pricing}.%

Most of the existing literature on caching \cite{poularakisexploiting,blaszczyszyn2014optimal,blasco2014content,ji2014fundamental} has focused on determining the optimal caching policy under various network scenarios. 
In \cite{poularakisexploiting}, the authors proposed a caching policy that optimizes the overall energy consumption of the SBSs while accounting for the multicast opportunities. The authors in \cite{blaszczyszyn2014optimal} proposed a caching policy that accounts for the geographical position of the SBSs and users. 
The work in \cite{blasco2014content} formulated the cache placement problem as a multi armed-bandit problem in which the SBSs do not know the popularity of the files. 
The authors in \cite{ji2014fundamental} proposed a coded caching strategy for wireless device-to-device (D2D) networks. 

Remarkably, none of these works have addressed the economic aspect of caching. Recently, the works in \cite{li2016pricing} and \cite{chen2016caching} provides some insights on the economics of caching using game-theoretic solutions. However, these works typically assume that the players are honest and have perfect knowledge of the cellular network information, which is not a reasonable assumption in practice. In fact, the CPs can modify the outcome of the incentive mechanism in their favor by attempting to manipulate the MNO's agreements. For example, some CPs can declare inaccurate information about their traffic load so as to mislead the MNO into charging them much lower prices. Such possibilities can therefore incentivize some of the CPs to not share their truthful private information in which case the results in \cite{li2016pricing} and \cite{chen2016caching} no longer hold.

The main contribution of this paper is to introduce a novel incentive mechanism for facilitating the deployment of caching in SCNs. We consider a model in which an MNO proposes agreements to the CPs to incite them to cache their content. We consider a practical scenario with \emph{asymmetric information} within which the CPs can be of different, private types that are unknown to the MNO. This private information pertains to the level of generated CP traffic which corresponds to the popularity of its content. The proposed approach, based on contract theory \cite{borgers2015introduction}, allows the MNO to define a contract for each CP in the presence of asymmetric information, by fixing a price for the allocated storage space to the CP. The goal of the MNO is to maximize the global reward of the CPs and cover the expenses of the caching process. Unlike classical contract-theoretic models \cite{gao2011spectrum,duan2014cooperative,zhang2015contract}, we show that the proposed model exhibits strategic interdependence between the CPs. Consequently, for this interdependent contract model, we derive the closed-form expression of the price and prove that the resulting mechanism satisfies the feasibility constraints of a contract. The designed mechanism is then shown to be budget balanced as it allows the MNO to offload its backhaul while covering the cost of caching the CPs' content. Simulation results show the performance advantage of the proposed scheme as well as its ability to incite truthfulness on the CPs.

\vspace{-0.2cm}
\section{System Model}
\label{model}
Consider a SCN composed of a set $\mathcal{M}$ of macro base stations (MBSs) and a set $\mathcal{S}$ of SBSs deployed by an MNO. The SBSs are connected to the MBSs via capacity-limited backhaul links and serve a set of users $\mathcal{U}$ that are subscribed to multiple CPs from a set $\mathcal{C}$. The SBSs can cache content to offload the MNO's backhaul. However, the MNOs needs the CPs' cooperation to cache their content. To motivate the CPs to participate in the caching, the MNO must offer contracts that present significantly improved QoS for the CPs' users. The contract terms between an MNO and a CP determine the price charged to the CP by the MNO and the amount of storage space offered.

We consider heterogeneous CPs with different traffic loads and content popularity. Based on this traffic load, the CPs have different incentive levels towards sharing their content with the MNOs. Naturally, there is an \emph{information asymmetry} between the MNO and the CPs. The CP is aware of its users' traffic as well as its preferences while the MNO may not have that information. Consequently, the CPs may have an incentive to not reveal their correct types so as to pay lower prices to the MNO. To overcome this challenge, the MNO must specify a suitable performance-reward \emph{bundle contract} $(\pi,\rho)$, where $\pi$ is the monetary reward that is paid by the CP, and $\rho$ is the storage space allocated to the CP. 

The goal of each CP $k\in\mathcal{C}$ is to maximize the performance of its users which depends on the amount of content that other CPs in the set $\mathcal{C}\setminus k$ will be willing to cache at the SBSs. In fact, the higher the traffic load of the CPs in $\mathcal{C}\setminus k$, more storage space is allocated to cache their content. Thus, less storage space can be made available for the CP $k$. This will negatively impact the data rate of CP $k$'s users as more requests need to be served via the capacity-limited backhaul. Moreover, by introducing caching, the traffic load of each SBS increases when caching highly popular files. Thus, more power is required at the SBS to serve all the requests for a highly popular file. Consequently, other CPs' users might experience large interference level from the SBSs that cache the files of CPs having a higher traffic load. 

\vspace{-0.1cm}
\subsection{Transmission Data Rate}
 The performance $r_k$ of a CP $k$ that results from caching its contents is measured by the transmission rate that its users experience from the serving MBS or SBS. When an SBS $i$ serves a user $j$, the data rate will be:

\small
\begin{equation}
\alpha_{ij}(\rho(\boldsymbol{\theta}),\boldsymbol{\theta})=\mathbb{E}_t\left[w_{ij}\log{\Big(1+\frac{p_{ij}(\rho,\theta_i)|h_{ij}|^2}{\sigma^2+I(\rho,\boldsymbol{\theta})}\Big)}\right],
\label{first_equ}
\end{equation}
\normalsize
where $I(\rho,\boldsymbol{\theta}) =\sum_{k\in\mathcal{S}\setminus i}{p_{kj}(\rho,\boldsymbol{\theta})|h_{kj}|^2}$ is the interference experienced by user $j$ from all the other SBSs. $w_{ij}$ is the channel bandwidth, $p_{ij}(\rho,\boldsymbol{\theta})$ is the transmit power from SBS $i$ to user $j$, $|h_{ij}|^2$ is the channel gain between SBS $i$ and user $j$, and $\sigma^2$ is the variance of the Gaussian noise. The vector $\boldsymbol{\theta}=[\theta_1,...,\theta_C]$ represents the traffic load of the CPs. Thus, the higher the traffic load of the CPs, the higher is the interference experienced by the users served from the neighboring SBSs. Since caching is done during off-peak periods, the transmit power is averaged over the considered time period.

If the file is cached at the SBSs, then the users will experience a relatively high data rate as the content is closer to them. However, if the data is not available at the associated SBS to serving a certain user, then the SBS must fetch the user content from the MBS over the capacity-limited and congested backhaul, yielding higher delays. The data rate of a user $j$ requesting file $f$ from its associated SBS $i$ can be given by:

\vspace{-0.5cm}
\begin{multline}
r_{ij}(\boldsymbol{\theta}) = (1-\beta_{if}(\rho(\boldsymbol{\theta}),\boldsymbol{\theta}))\min{\{\alpha_{ij}(\rho(\boldsymbol{\theta}),\boldsymbol{\theta}),\alpha^{\prime}_{mi}\}}\\+\beta_{if}(\rho(\boldsymbol{\theta}),\boldsymbol{\theta})\alpha_{ij}(\rho(\boldsymbol{\theta}),\boldsymbol{\theta}),
\end{multline}
\normalsize
 where $\boldsymbol{\beta}\in \{0,1\}^{S\times F_k}$ is the outcome of the MNO's storage allocation $\rho$ and $F_k$ is the cardinality of the set of files $\mathcal{F}_k$ provided by a CP $k$.  $\boldsymbol{\beta}$ depends on the caching policy $\rho$ and the traffic load of the CPs. For instance, when a CP $k$ has highly popular files or its willingness level to cache its content is high, it will impact the storage allocation to the other CPs $\mathcal{C}\setminus k$. The larger the number of files that CP $k$ wants to cache, the lower is the storage space that will be allocated to other CPs $\mathcal{C}\setminus k$ and vice versa. Each entry $\beta_{if}$ is a binary variable that equals $1$ if file $f$ is cached at SBS $i$ and $0$ otherwise. $\alpha^{\prime}_{mi}$ is the data rate from the MBS $m$ to SBS $i$ and is given by:

\begin{equation}
\alpha^{\prime}_{mi}=\mathbb{E}_t\left[w_{mi} \log{\Big(1+\frac{p_{mi}|h_{mi}|^2}{\sigma^2+I^{\prime}}\Big)}\right],
\label{second_equ}
\end{equation}
\normalsize
where $I^{\prime}=\sum_{l\in\mathcal{M}\setminus m}{p_{li}|h_{li}|^2}$ is the interference experienced by SBS $i$ from all the other transmitting MBSs.

Thus, the total rate of the users of CP $k$ can be given by:

\begin{equation}
r_k(\rho(\boldsymbol{\theta}),\theta_k) = \sum_{i\in\mathcal{S}}{\sum_{j\in\mathcal{U}_{ki}}{r_{ij}(\boldsymbol{\theta})}},
\end{equation}
\normalsize
where $\mathcal{U}_{ki}\subseteq\mathcal{U}_k$ is the set of users that request at least one file from CP $k$ by using SBS $i$, and $\mathcal{U}_k$ is the set of users requesting files of CP $k$. Here, we note that, in our model, each CP will have private information that is modeled as a \emph{type of CP} as discussed next.


\subsection{Content Provider Type}

We define the CP's type to be a representation of its traffic load and content popularity. In fact, when the MNO offers a contract, it must account for the generated traffic by the CPs. For example, by caching the contents of CPs with a high traffic, the MNO can serve more requests locally, thus decreasing its backhaul load considerably. Here, we consider that the number of CP types belongs to a discrete, finite space and grouped as follows:

\begin{definition}
\emph{There are $C$ CPs that generate traffic over an MNO's network. The CPs' types are sorted in an ascending order and classified into $K$ types $\theta_1,...,\theta_K$ with $K\leq C$. Each type includes properties such as the willingness to cache and the global popularity of the CPs files. The types are ordered as follows: $\theta_1<...<\theta_k<...<\theta_K,  ~~~k\in \{1,...,K\}.$}
\end{definition}

Since the types are not known by the MNO, the CPs can announce wrong information about their types so that they improve the performance of their users. For example, by claiming that its content popularity is higher than it actually is, a CP $k$ can mislead the MNO to allocate more storage space. In such a case, the CP can end up paying lower prices while also lowering the interference experienced by its users. Indeed, the truthful popularity information is necessary for the MNO to define the contracts that optimize the benefit of the CPs and cover the implementation costs of caching. Such cost includes the expenses of deploying storage devices and the required power to download the content and refresh the storage units. Here, our goal is to design contracts that incentivize the CPs to reveal the true values of their types $\boldsymbol{\theta}$ to the MNO. To this end, the contracts will be designed such that no CP can profit by choosing a contract that is designed for other types.


\vspace{-0.2cm}
\subsection{Content Provider Model}
The utility function of a CP $k$ of type $\theta_k$ that decides to cache a set of files $\mathcal{F}_k$ at the operator's network is:
\begin{equation}
u_{k}(\boldsymbol{\theta})=  r_k(\rho_k(\boldsymbol{\theta}),\theta_k)- \pi_{k}(\boldsymbol{\theta}) ,
\end{equation}
where $r_k(\rho)$ is defined in (4) and represents the valuation function regarding the rewards, which is a strictly increasing concave function of $\rho_k$, with $r(0)=0$ and $r^{\prime}(\rho_k)>0$, $r^{\prime\prime}(\rho)<0$ for all $\rho_k$. $\pi_k$ represents the price charged by the MNO for a storage allocation $\rho_k$.


\vspace{-0.1cm}
\subsection{Mobile Network Operators Model}
By caching the content of the CPs, the MNO will be able to reduce the traffic load on its backhaul. This benefit depends on the traffic load of the CPs as dictated by the popularity of their cached traffic. Thus, the MNO will generally prefer to cache the most popular files. By doing so, for the same storage capacity, the load can be reduced more for a CP whose files have a high popularity compared to other CPs. Thus, the cost of storage $c_s$ at the MNO can be given as a function of the traffic load of the considered CP as, $c(\theta)=\log{(1+\theta)}$. This storage cost function increases quickly up to a certain threshold and then increases slowly. It is suitable to model the storage cost as the MNO must allocate more storage space to serve a given traffic load, and this cost becomes insignificant when the traffic load increases as some requests become redundant.

A proper utility function for the MNO can be defined as the monetary reward that is charged to the CPs minus the cost of the allocated resources by the MNO, including storage.
\begin{equation}
v_{k}(\boldsymbol{\theta})= \pi_{k}(\boldsymbol{\theta})-c_k(\boldsymbol{\theta},\theta_k),
\end{equation}   
where $\pi_k$ is the price that the operator charges CPs of type $k$.
The total expected utility of the operator can be given by
\begin{equation}
v =\sum_{k\in\mathcal{C}} v_{k}(\boldsymbol{\theta}).
\end{equation}

In the considered model, the MNO is assumed to get the CP's types directly from the CPs. Based on this information, the goal of the MNO is to determine a contract for all possible CP types that maximizes the global benefit of the CPs. At the same time, the CP ensures that its utility is nonnegative by making the prices charged to the CPs to at least cover its cost. This optimization problem can be defined as follows:
\begin{equation}\label{PF1}
\begin{aligned}
& \underset{(\pi_k,\rho_k)}{\max}
& & \sum_{k\in\mathcal{C}}u_k(\rho_k(\boldsymbol{\theta}),\theta_k) \\
& \text{subject to}
& & v \geq 0.
\end{aligned}
\end{equation}
In this formulation, we do not make any constraint on the participation of the CPs. Thus, when proposing the contracts resulting from solving (\ref{PF1}), CPs may prefer not to select any of the contracts or select contracts that are not designed for their types. To analyze this economic incentive problem, next, we propose a solution based on the framework of \emph{contract theory} for designing feasible contracts \cite {borgers2015introduction}.
\section{Proposed Incentive Mechanism for Caching}
\label{form}
We consider a contract-theoretic problem in which an MNO seeks to motivate the CPs to participate and help it in the deployment of caching solutions. Due to the limited storage capacity of the SBSs as well as the limited capacity of the access links between the SBSs and the users, the allocation of storage space and power has an impact on the utilities of the CPs. In fact, the QoS achieved by the CPs' users depends on the interference from the other SBSs as shown in (\ref{first_equ}) and (\ref{second_equ}). Moreover, the allocated storage capacity to a given CP depends on the number of CPs that have signed contracts with the same MNO. For instance, the more storage is allocated to a CP $k$ the lower is the available storage capacity for other CPs. Thus, in the considered model, there is \emph{interdependence} between the signed contracts  by the different agents.

Classical contract theory models that are used to model resource allocation problems in wireless networks such as in \cite{gao2011spectrum,duan2014cooperative,zhang2015contract} cannot be applied for the analysis of caching incentive problem between CPs and an MNO defined in (\ref{PF1}). In fact, these works assume that the contract selected by a CP does not impact the utility of other CPs or focus only on models with one MNO and one CP. Thus, none of the existing works account for the interactions between the CPs. Moreover, the revelation of misleading information by a given CP in a multiple CPs model not only impacts the MNO but also impacts other CPs, which is not considered in \cite{gao2011spectrum,duan2014cooperative,zhang2015contract}.

To define the most appropriate contract for the formulated problem (\ref{PF1}), we consider the so-called truthful dominant strategy implementation. Under such contracts, the solution that maximizes the utility function of the CPs will require those CPs to reveal their private information which, in our model, pertains to the real popularity of their content. The goal of the MNO is to maximize the social welfare which effectively captures the global QoS that is experienced by the users of all CPs. Moreover, the MNO ensures that the cost of serving these users is at least covered by the price charged to the CPs. To incite the CPs to collaborate with the MNO via caching, the contract that a CP selects must be \emph{feasible} in that it satisfies the following necessary and sufficient constraints:

\begin{definition}
 Ex-post Individual Rationality (IR): \emph{The contract that a CP selects should guarantee that the utility of the CP is nonnegative for any $\boldsymbol{\theta}_{-k}$ declared by the other CPs,\vspace{-0.3cm}
\begin{multline}\label{IR}
 r_k(\rho_k(\theta_k,\boldsymbol{\theta}_{-k}),\theta_k)- \pi_{k}\geq 0, ~ \forall k \in \{1,...,K\}.
\end{multline}}
\end{definition}  

\begin{definition}
Incentive Compatibility (IC): A\emph{ contract satisfies incentive compatibility constraint if each CP of type $\theta_k$ prefers to reveal its real type $\theta_k$ rather than another type $\hat{\theta}_k$, i.e., \vspace{-0.3cm}
\begin{multline}\label{IC}
 r_k(\rho_k(\theta_k,\boldsymbol{\theta}_{-k}), \theta_k)- \pi_{k}\geq  r_k(\rho_k(\hat{\theta}_k,\boldsymbol{\theta}_{-k}),\theta_k)- \pi_{k}.
\end{multline}}
\end{definition}
\subsection{Incentive Mechanism Analysis}
The goal of the MNO is to determine a pricing policy that motivates the CPs to declare their real type and simultaneously participate in the caching system through a budget balanced mechanism, i.e., the MNO would not experience a negative utility and its effort is covered by the price charged to the CPs. To this end, the CPs need to declare their types for the MNO that in turn optimizes their utility while accounting for the necessary conditions for contracts feasibility. The optimization problem of the MNO can be defined as follows:

\small
\begin{equation}\label{PF}
\begin{aligned}
& \underset{(\pi_k,\rho_k)}{\max}
& & \sum_{k\in\mathcal{C}}u_k(\rho_k(\boldsymbol{\theta}),\theta_k) \\
& \text{subject to}
& & (\ref{IR}),(\ref{IC}),  v \geq 0.
\end{aligned}
\end{equation}
\normalsize

The solution of this problem consists in the determination of the components of a contract that consist in the allocated storage space and the price charged to each CP. The closed-form of the contract is provided by the following theorem.
\begin{theorem}
\emph{The unique efficient solution of the optimization problem (\ref{PF}) can be given by:}

\small
\begin{equation}
\rho_k^* \in \arg\max\limits_{\rho_k}\sum_{i}\left[ r_i (\rho_i(\boldsymbol{\hat{\theta}}),\hat{\theta}_i)-c_i(\boldsymbol{\hat{\theta}})\right], \forall k,
\end{equation}
\vspace{-0.3cm}
\begin{multline}\label{sol}
\pi_k(\boldsymbol{\hat{\theta}})=\underbrace{\Big[\max\limits_{\rho_i} \sum_{i\neq k} r_i(\rho_i(\boldsymbol{\hat{\theta}}_{-k}),\hat{\theta}_i)-c_i(\boldsymbol{\hat{\theta}}_{-k})\Big]}_{(a)}\\ -\Big[\underbrace{\sum_{i \neq k} r_i(\rho_i^*( \boldsymbol{\hat{\theta}}),\hat{\theta}_i)-c_i(\boldsymbol{\hat{\theta}})\Big]}_{(b)},
\end{multline}
\normalsize
\emph{where (a) represents the maximized social welfare when CP $k$ is not considered while in (b), CP $k$ is considered. Moreover, $\hat{\boldsymbol{\theta}}$ represents the revealed type by the CPs while $\boldsymbol{\theta}$ is the real type of the CPs.}
\end{theorem}
\begin{proof}
The proof is provided in the Appendix.
\end{proof}
This result shows that, in order to determine the terms of a contract with a CP $k$, the MNO first, allocates the storage space to CP $k$ by solving the optimization problem (12). It is clear that the problem in (12) is NP-hard and thus it is challenging to find the optimal storage allocation. To solve (12), we use the framework of matching theory to analyze the assignment of storage space between the MNO that acts on behalf of its SBSs and the CPs \cite{gale1962college}. Matching theory is a suitable framework to solve NP-hard assignment problems such as in (12). As stated before, the allocated storage space to a CP $k$ depends on the allocated storage to the other CPs which is known as externalities. We solve the problem using a swap-based deferred acceptance algorithm which is guaranteed to converge to a stable outcome \cite{gale1962college}. Each CP $i$ starts by requesting from the MNO, a given storage space $\rho_i$ that maximizes $r_i(\rho_i(\boldsymbol{\hat{\theta}}_{-k}),\hat{\theta}_i)-c_i(\boldsymbol{\hat{\theta}}_{-k})$. After receiving all the requests from the CPs and based on the caching policy used by the MNO, it defines the accepted requests that maximize (12) and rejects the others. This is defined while accounting for the limited storage capacity of its SBSs. If the request of a given CP is rejected, the CP decreases the amount of the storage space it requests from the MNO. The MNO accepts new requests and rejects others based on the allocation configuration that maximizes (12). The procedure is repeated until there does not exist a CP $k$ that prefers to be assigned a given storage capacity $\rho_k$ and this allocation also maximizes (12) at the MNO.  

Once the storage space is allocated to the CPs, the price paid by a CP $k$ is found from (13), which accounts for the impact of CP $k$ on the utility of other CPs. This price represents the difference between the global utility achieved by all the CPs when CP $k$ participates in the caching process, and the global utility achieved by the CPs when CP $k$ does not participate. Note that a CP $k$ can impact other CPs utilities in two ways. The first one is through the allocated storage space. In fact, when more storage space is allocated to CP $k$, less storage is available for other CPs and thus more requests of these CPs are served via the backhaul. The second is the traffic load of CP $k$ as the transmit power of the SBSs increases by increasing the number of served requests for that CP's files. Thus, we can deduce that higher traffic load of a CP $k$ and large amounts of allocated storage to CP $k$ will result in an increase in the price charged by the MNO. The dependence of the price on the traffic load of the CPs, i.e., CPs type, is given next.

\begin{corollary}
\emph{When $\theta_k\geq\theta_l$ then we have $\pi_k\geq\pi_l$.}
\end{corollary}
\begin{proof}
This results follows directly from the monotonicity property of the rate function and the structure of (13).
\end{proof}
\vspace{-0.3cm}
\section{Simulation Results}
\label{sim}
\vspace{-0.1cm}

\begin{figure*}[t]
\subfloat[ \label{res1}]{%
	\includegraphics[width=0.33\linewidth]{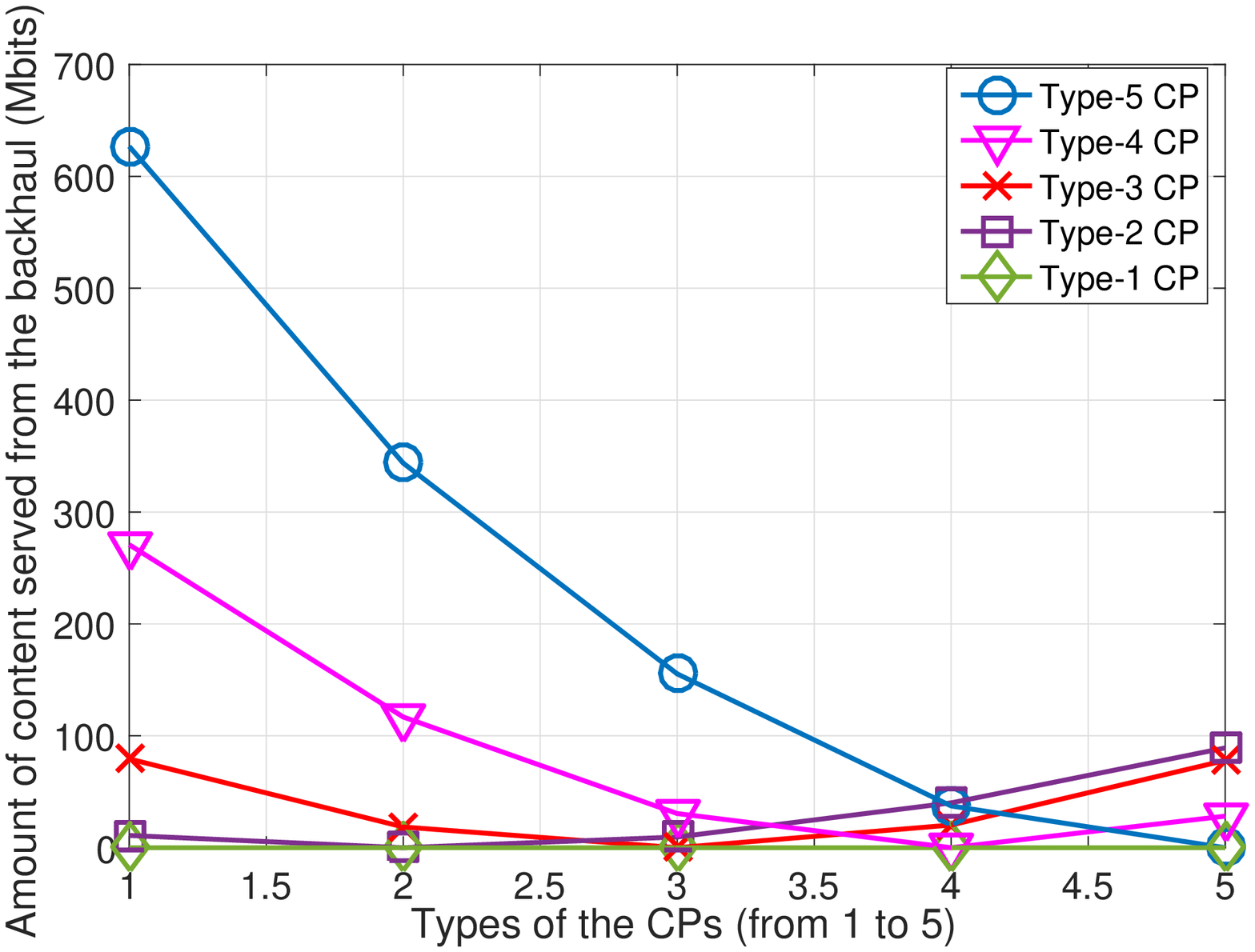}
}
\subfloat[ \label{res2}]{%
	\includegraphics[width=0.33\textwidth]{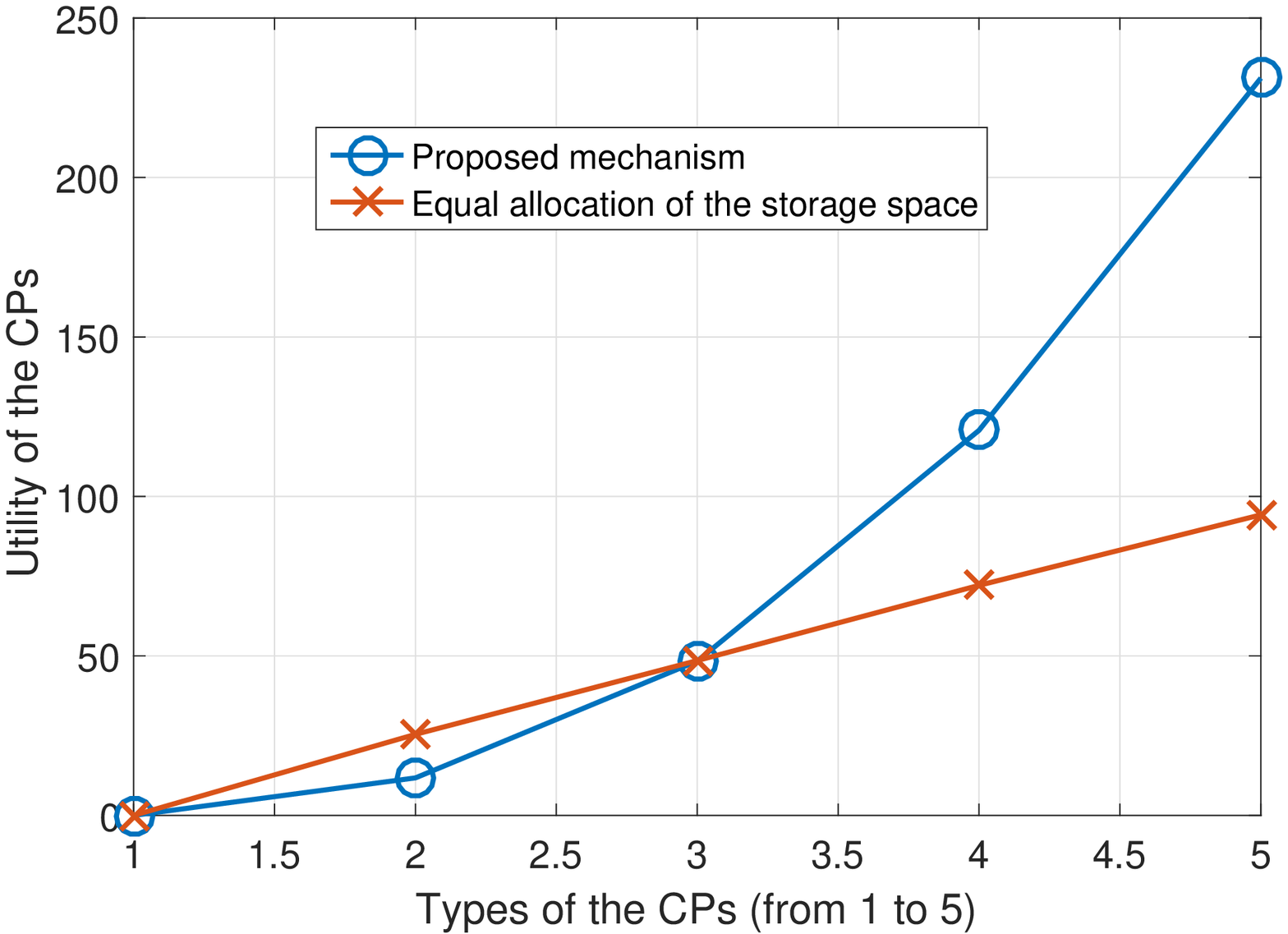}
}
\subfloat[ \label{res3}]{%
	\includegraphics[width=0.33\linewidth]{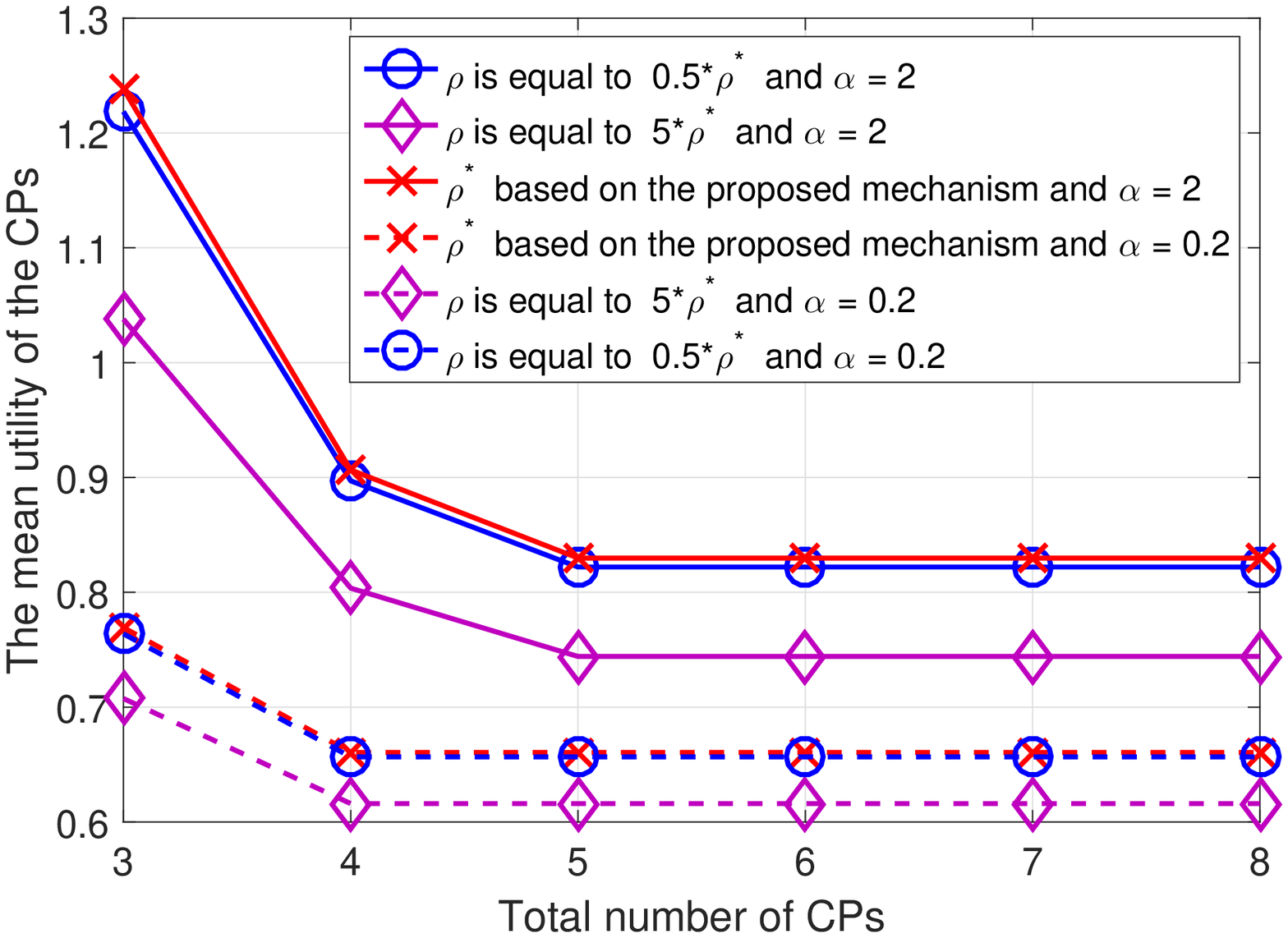}
}
\caption{\small Numerical results for a) the amount of content served via the backhaul with respect to each CP's type, b) the utility of the CPs as a function of the CPs' type and the used storage allocation model, c) the mean utility of a CP with respect to the total number of CPs and the popularity distribution of the files. \vspace{-0.5cm}}
\label{res}
\end{figure*}

For our simulations, we consider five CPs with different traffic load levels from 1 to 5 with type 5 being the highest load. A type-1 CP is chosen with no traffic and is used as a baseline to compare the performance of our mechanism with the case in which there is no caching. We consider a set of 100 files whose popularity follows a Zipf distribution of parameter $\alpha=0.2$. The MNO has one MBS that serves all the requests that cannot be served from the SBSs' cache. The number of SBSs is 10 and the total storage capacity of the SBSs is 1 Gbits. The transmit power of the SBSs is 1 W and the bandwidth capacity to 100 MHz. 

In Fig. \ref{res1}, we show that the amount of content that is served via the backhaul for every CPs' users when a CP selects the contracts designed by the MNO for each type. We account for the fact that the price that can be paid by each CP is limited and the limit increases by increasing the type of the CPs. From Fig. \ref{res1}, we can observe that when high type CPs select contracts that are designed for CPs with lower traffic load, the amount of content that is served via the backhaul increases until it reaches the maximum which corresponds to the lowest type contract. Similarly, when low type CPs select the contracts designed for higher types, the amount of traffic served from the backhaul increases. This is due to the high charged price by the MNO for high type CPs and thus CPs of lower type cannot afford that which results in a larger amount of served content from the backhaul. Thus, this result validates the fact the proposed approach for cache incentive compatible as is forces each CP to choose the contract designed for its own type.


In Fig. \ref{res2}, we compare the case in which all the CPs choose the contracts defined for their corresponding types by the MNO and the case in which the storage space is allocated equally by the MNO for all the CPs. Fig. \ref{res2} show that the model of equal storage allocation outperforms the proposed mechanism for low type CPs. In this case, the allocated storage space for the CPs is higher than the required storage by the CPs which appears through the utilities of type-2 CP and type-3 CP that are only 2\% to 10\% higher than their utilities when following the proposed mechanism. On the other hand, the proposed mechanism outperforms the equal storage space allocation model for high type CPs. In fact, the utility of high level types is higher when selecting the offered contracts by the MNO for their specific types. The utilities of type-4 and type-5 CPs are 50\% to 140\% higher than their achievable utilities in the model of equally allocated storage space. The proposed mechanism is more beneficial for the CPs as it allows the MNO to offer to all CPs only the amount of storage space they need. In contrast, when using the equal allocation approach, the MNO allocates to the low type CPs more than their storage requirements and insufficient space for high level type CPs.

In Fig. \ref{res3}, we show the variation of the mean utility of the CPs when increasing the total number of CPs. We consider two different values for the parameter $\alpha$, $0.2$ which corresponds to the case in which the files have comparable popularity, and $2$ for the case in which some files are very popular while others have a very low popularity. We can observe that the mean utility for the CPs decreases by increasing the total number of CPs in the model. This is due to the data rate function that depends on the additional interference from the added CPs and the decrease of the available storage space that can be allocated to each CP. The achievable utility by a CP is up to 20\% larger compared to the cases in which the CPs request a storage capacity that is larger or lower than the offered one by the MNO. Moreover, we can see that the popularity of the files impacts the mean utility of the CPs. In fact, the CPs can achieve a larger utility when the distribution of the popularity of the files is steep and thus by caching a file, a large amount of the requests can be served from the cache of the SBSs.
\vspace{-0.1cm}
\section{Conclusion}
\label{conc}
\vspace{-0.1cm}
In this paper, we have proposed a new incentive framework to motivate the CPs to cooperate with an MNO and cache their content at the MNO's SBSs. Based on contract theory, we have designed an incentive mechanism that allows the MNOs to offer a contract for each CP in which it sets the allocated storage for the CP and the charged price by the MNO for the caching service. This model accounts for both asymmetry of information and the interdependence between the different contracts. We have then derived the optimal pricing mechanisms and contracts that motivate the CPs to cache their content and reveal their private information. Simulations have shown the effectiveness of the proposed approach in inciting the participation of CPs in caching.

\begin{appendix}
\vspace{-0.cm}
We prove Theorem 1 by showing the following Lemmas.
\begin{lemma}
\emph{The proposed mechanism is incentive compatible.}
\end{lemma}
\begin{proof}
We show this result by contradiction. Suppose that $\rho$ is an efficient decision rule but $(\rho,\pi)$ is not dominant strategy incentive compatible. Then, there exists $i, \boldsymbol{\theta},$ and $\boldsymbol{\hat{\theta}}$ such that:
\begin{equation*}
r_i(\rho_i(\hat{\theta}_i,\boldsymbol{\theta}_{-i}),\theta_i)-\pi_i(\boldsymbol{\theta}_{-i},\hat{\theta}_i)>r_i(\rho_i(\boldsymbol{\theta}),\theta_i)-\pi_i(\boldsymbol{\theta}).
\end{equation*}
From (\ref{sol}), this implies that
\begin{multline*}
r_i(\rho_i(\hat{\theta}_i,\boldsymbol{\theta}_{-i}),\theta_i)-(a)>r_i(\rho_i(\boldsymbol{\theta}),\theta_i)-(b),
\end{multline*}
which is equivalent to

\vspace{-0.5cm}
\small
\begin{multline*}
r_i(\rho_i(\hat{\theta}_i,\boldsymbol{\theta}_{-i}),\theta_i)-\sum_{j\neq i}r_j(\rho_j(\hat{\theta}_i,\boldsymbol{\theta}_{-i}),\theta_j)-c_i(\hat{\theta}_i,\boldsymbol{\theta}_{-i})>\\r_i(\rho_i(\boldsymbol{\theta}),\theta_i)-\sum_{j\neq i}r_j(\rho_j(\theta_i,\boldsymbol{\theta}_{-i}),\theta_j)-c_i(\theta_i,\boldsymbol{\theta}_{-i}).
\end{multline*}
\normalsize
This contradicts the efficiency of $\boldsymbol{\rho}$ based on (12) and thus, the assumption was incorrect.
\end{proof}

\begin{lemma}
\emph{Truth telling is a dominant strategy under (\ref{sol}).}
\end{lemma}
\begin{proof}
Consider the problem of choosing the best type $\hat{\theta}_i$ by a CP $i$. A best strategy for CP $i$ solves 
\begin{equation*}
\begin{aligned}
& \underset{\hat{\theta}_i}{\max}
& & r_i(\rho_i(\hat{\boldsymbol{\theta}}),\hat{\theta}_i)-\pi_i(\hat{\boldsymbol{\theta}}).
\end{aligned}
\end{equation*}
Substituting the payment function by the proposed mechanism (12), we get

\vspace{-0.5cm}
\small
\begin{multline*}
 \underset{\hat{\theta}_i}{\max} \Big[r_i(\rho_i(\hat{\boldsymbol{\theta}}),\hat{\theta}_i)-\underbrace{\sum_{j\neq i} r_j(\rho(\hat{\boldsymbol{\theta}}_{-i}),\theta_{j})-c_j(\theta_j,\hat{\boldsymbol{\theta}}_{-i})}_{(a)}\\+\sum_{j\neq i}r_j(\rho_j(\hat{\boldsymbol{\theta}}),\hat{\theta}_j)\Big].
\end{multline*}
\normalsize
Since $(a)$ does not depend on $\hat{\theta}_i$, it is sufficient to solve 

\vspace{-0.3cm}
\small
\begin{equation*}
\begin{aligned}
& \underset{\hat{\theta}_i}{\max}
& & \bigg(r_i(\rho_i(\hat{\boldsymbol{\theta}}),\hat{\theta}_i)+\sum_{j\neq i}r_j(\rho_j(\hat{\boldsymbol{\theta}}),\hat{\theta}_j)\bigg).
\end{aligned}
\end{equation*}
\normalsize
Thus, CP $i$ would pick a declaration $\hat{\theta}_i$ that will lead the mechanism to pick a $\boldsymbol{\rho}$ which solves

\vspace{-0.3cm}
\small
\begin{equation}\label{eq}
\begin{aligned}
& \underset{\boldsymbol{\rho}}{\max}
& & \bigg(r_i(\rho_i(\hat{\boldsymbol{\theta}}),\hat{\theta}_i)+\sum_{j\neq i}r_j(\rho_j(\hat{\boldsymbol{\theta}}_{-i}),\hat{\theta}_j)\bigg).
\end{aligned}
\end{equation}
\normalsize
Under the proposed mechanism,

\vspace{-0.3cm}
\small
\begin{equation*} 
\boldsymbol{\rho}^*\in \arg\max\limits_{\boldsymbol{\rho}}~~\bigg(r_i(\rho_i(\hat{\boldsymbol{\theta}}),\theta_i)+\sum_{i\neq j}r_j(\rho_j(\hat{\boldsymbol{\theta}}),\hat{\theta}_j)\bigg).
\end{equation*}
\normalsize
The proposed mechanism (\ref{sol}) will choose $\rho$ in a way that solves the maximization problem (\ref{eq}) with $\hat{\theta}_i=\theta_i$. Thus, truth-telling is a dominant strategy for CP $i$.
\end{proof}

\begin{lemma}
\emph{The proposed mechanism (\ref{sol}) is ex-post individually rational.}
\end{lemma}
\begin{proof}
At the equilibrium, all the CPs are truthful and declare their real types. Thus, by replacing (\ref{sol}) in the utility of a CP $i$, we have

\vspace{-0.3cm}
\small
\begin{multline}\label{IRP}
u_i=\sum_{i} r_i(\rho_i^*(\boldsymbol{\theta}),\theta_i)-c_i(\boldsymbol{\theta})-\sum_{j\neq i}r_j(\rho_j^*(\boldsymbol{\theta}_{-i}),\theta_j)-c_j(\boldsymbol{\theta}),
\end{multline}
\normalsize
where $\boldsymbol{\rho}^*$ is the outcome that maximizes the social welfare. The CPs could have picked $\rho_i(\boldsymbol{\theta}_{-i})$ instead of $\rho_i(\boldsymbol{\theta})$ as a solution of the optimization problem, as it is one of the possible strategies. This is possible because the set of strategies, i.e., the total storage capacity of the MNO, is fixed and does not change by changing the set of participating CPs. Thus, 

\vspace{-0.2cm}
\small
\begin{equation*}
\sum_{j}r_j(\rho_j^*(\boldsymbol{\theta}),\theta_j)\geq \sum_{j}r_j(\rho_j^*(\boldsymbol{\theta}_{-i}),\theta_j).
\end{equation*}
\normalsize
Furthermore, we know that the rate of a participating CP cannot be negative, i.e.,

\vspace{-0.4cm}
\small
\begin{equation*}
r_j(\rho_j^*(\boldsymbol{\theta}_{-i}),\theta_j)\geq 0.
\end{equation*}
\normalsize
Therefore,

\vspace{-0.3cm}
\small
\begin{equation*}
\sum_{i} r_i(\rho_i^*(\boldsymbol{\theta}),\theta_i)\geq\sum_{j\neq i}r_j(\rho_j^*(\boldsymbol{\theta}_{-i}),\theta_j).
\end{equation*}
\normalsize
Thus, (\ref{IRP}) is non-negative and the proposed mechanism is ex-post individual rational.
\end{proof}

\begin{lemma}
\emph{The proposed mechanism (\ref{sol}) is weakly budget-balanced.}
\end{lemma}
\begin{proof}
Since the CPs are truth-telling at the equilibrium then we have

\vspace{-0.6cm}
\small
\begin{multline*}
\sum_{i} \pi_i(\boldsymbol{\theta})= \sum_{i}\bigg(\Big[\sum_{j\neq i}r_j(\rho_j(\boldsymbol{\theta}_{-i}),\theta_j)-c_j(\boldsymbol{\theta}_{-i},\theta_j)\Big]\\-\Big[\sum_{j\neq i}r_j(\rho_j(\boldsymbol{\theta}),\theta_j)-c_j(\boldsymbol{\theta},\theta_j)\Big]\bigg).
\end{multline*}
\normalsize
Moreover, since the utility of a CP is a decreasing function of the number of the set of participating CPs, we have that, $\forall i$,

\vspace{-0.3cm}
\small
\begin{equation*}
\sum_{j\neq i}r_j(\rho_j(\boldsymbol{\theta}_{-i}),\theta_j)-c_j(\boldsymbol{\theta}_{-i},\theta_j)\geq\sum_{j\neq i}r_j(\rho_j(\boldsymbol{\theta}),\theta_j)-c_j(\boldsymbol{\theta},\theta_j).
\end{equation*}
\normalsize
Thus, the proposed mechanism is weakly budget-balanced.
\end{proof}
Next, we prove the provided result in Theorem 1.
\begin{proof}
Based on Lemma 1, Lemma 3 and Lemma 4, we can deduce that all the condition of the optimization problem (\ref{PF}) are satisfied. Based on Lemma 2, we can deduce that the proposed mechanism (\ref{sol}) is the unique efficient solution of the formulated problem (\ref{PF}).
\end{proof}
\end{appendix}
\scriptsize 
\bibliographystyle{IEEEtran}
\bibliography{references}

\end{document}